\documentclass[11pt,english]{scrartcl}
\usepackage[T1]{fontenc}
\usepackage[latin9]{inputenc}
\usepackage{geometry}
\usepackage{babel}
\usepackage{amsthm}
\usepackage{amsmath}
\usepackage{amssymb}
\usepackage{esint}
\usepackage{color}
\usepackage{stmaryrd}
\usepackage[unicode=true,
 bookmarks=true,bookmarksnumbered=true,bookmarksopen=true,bookmarksopenlevel=5,
 breaklinks=false,pdfborder={0 0 0},backref=false,colorlinks=false]
 {hyperref}
\hypersetup{pdftitle={Short-term Uncertainty},
 pdfauthor={Arturo Valdivia},
 pdfsubject={Credit Risk}}

\usepackage{float}
\usepackage{graphicx}
\usepackage{epstopdf}
\makeatletter
\theoremstyle{plain}
\newtheorem{thm}{\protect\theoremname}[section]
  \theoremstyle{remark}
  \newtheorem{rem}[thm]{\protect\remarkname}
  \theoremstyle{plain}
  \newtheorem{lem}[thm]{\protect\lemmaname}
  \theoremstyle{plain}
  \newtheorem{prop}[thm]{\protect\propositionname}
  \theoremstyle{definition}
  \newtheorem{example}[thm]{\protect\examplename}

\makeatother

  \providecommand{\examplename}{Example}
  \providecommand{\lemmaname}{Lemma}
  \providecommand{\propositionname}{Proposition}
  \providecommand{\remarkname}{Remark}
\providecommand{\theoremname}{Theorem}

\begin{document}

\title{CoCos under short-term uncertainty\thanks{To appear in \emph{Stochastics}. Final version will be available at http://dx.doi.org/10.1080/17442508.2016.1149590}}

\author{Jos\'e Manuel Corcuera%
\thanks{University of Barcelona, Gran Via de les Corts Catalanes, 585, E-08007
Barcelona, Spain. E-mail: \texttt{jmcorcuera@ub.edu}.%
}, \and Arturo Valdivia%
\thanks{Universitat de Barcelona, Gran Via de les Corts Catalanes, 585, E-08007
Barcelona, Spain. E-mail:\texttt{arturo.valdivia@ub.edu}%
}.}

\maketitle

\begin{abstract}
In this paper we analyze an extension of the Jeanblanc
and Valchev \cite{JeaValc} model by considering a short-term
uncertainty model with two noises. It is a
combination of the ideas of Duffie and Lando \cite{DuffieLando} and
Jeanblanc and Valchev \cite{JeaValc}:\ share quotations of the firm
are available at the financial market, and these can be seen as noisy
information about the fundamental value, or the firm's asset, from which a
low level produces the credit event. We assume there are also reports of the
firm, release times, where this short-term uncertainty
disappears. This
credit event model is used to describe conversion and default in a CoCo bond.
\vspace{1cm}

\end{abstract}

\section{\textrm{\textmd{\normalsize{}\label{sec: short term uncertainty}}}Introduction}


A contingent convertible (CoCo) is a bond issued by a financial institution
such that, upon the appearance of a credit event, an automatic conversion into a predetermined number of shares
takes place.  This credit event is related to a possible distress period of the
institution, and thus the conversion intends to be a loss absorbing security, in the sense that in
case of liquidity difficulties it produces a recapitalization of the entity. An alternative to the bond conversion into shares is to force a partial write-down of the bond's face value. We shall however only consider the first type of conversion here.

There is disagreement about how to establish the trigger or credit event. It
is perhaps the most controversial issue in a CoCo. Some advocate conversion
based on book values, such as the different capital ratios used in Basel III.
Others defend \emph{market triggers} such as the market value of the equity. So far
the CoCos issued by the private sector are based on accounting ratios.

From a modelling point of view, and depending on the trigger chosen for the
conversion, one can follow different approaches: a structural approach which links conversion to
a low level of a certain index related to the firm's asset, debt or equity; or a reduced-form approach assuming a conversion intensity that
can depend on certain explicative factors. This latter approach is specially
useful when \emph{pricing} CoCos is the main interest, it is a kind of
statistical modelling of the conversion trigger. In fact what one models is
the law of the conversion time. On the other hand, from the structural approach for modelling the
trigger one models the random variable describing the conversion time, and
one relates it to the dynamics of the firm's assets, debt, or equities. It is a
more explanatory approach where one can use the observed dynamics of
certain economics facts in order to explain the conversion time.

Conversion is a kind of credit event which is similar to default and
empirical facts show that such events happen all of a sudden, in
such a way that bond prices drop precipitously in a way not completely expected.
 This behavior, that is compatible with the existence of an
intensity for the credit event, is observed from the fact that yield spreads
are strictly positive at zero maturity. Structural models where the default
is linked to the first time that a diffusion process falls below a certain
level do not reproduce this behavior. Their corresponding default times are
predictable, the yield spreads go to zero when the time to maturity goes to
zero, they are models with zero intensity. Of course, we can consider a
jump-diffusion but then there is not much difference when modelling the
jumps of the diffusion from modelling the intensity, see, for instance, Zhou \cite{zhou2001term}. Another way of explaining the inaccessibility  of the default time is
that the information about the process involved in the default or conversion
is not perfect, maybe because the firm's assets are not observable, or they are
observable only at certain times, maybe because there is certain delay or
noise in the accounting reports, etc. Note that this is in fact the most
sensible and explanatory way of modelling the credit events. Credit events
have their own reasons linked to the economic behavior of the firm or one
sector, but at the same time only experts or managers are aware about
the current behavior of the firm or sector and market participants have
only noisy or delayed information. As a result the model, being structural
from the point of view of managers becomes a reduce form model from the
market view. A precedent for this approach is Duffie and Lando \cite{DuffieLando}
who is the first serious attempt to study the consequences of
noisy information in structural models. Another precedent is Jarrow et al. \cite{CetinEtAl}, where the market
information is a reduction of the manager's information, in particular they
assume that a credit event happens after a long negative cash balance
situation followed by a drop that duplicates the bad cash balance, however
market only sees the sign of cash balances and the time of default. This latter model is
certainly a very particular one that takes advantage of the behavior of
Brownian excursions but its extension, though appealing, is, as far as we
know, an open problem. Finally the third important reference is Jeanblanc
and Valchev \cite{JeaValc} where they discuss the effect of different types of
partial information, allowing some update times at which the information becomes complete.
This motivates the idea of what we call a \emph{short-term uncertainty model}, that is, a model
where the uncertainty in the observation is only present in between update times.

In this paper we analyze an extension of the Jeanblanc
and Valchev \cite{JeaValc} model by considering a short-term
uncertainty model with two noises. It is a
combination of the ideas of Duffie and Lando \cite{DuffieLando} and
Jeanblanc and Valchev \cite{JeaValc}:\ share quotations of the firm
are available at the financial market, and these can be seen as noisy
information about the fundamental value, or the firm's asset, from which a
low level produces the credit event. We assume there are also reports of the
firm, release times, where this short-term uncertainty
disappears. This
credit event model is used to describe conversion and default in a CoCo bond.

\section{Pricing CoCos}
Every pricing model for a CoCo starts by defining the mechanisms that
trigger default and conversion. From a structural modelling approach
we can say, in general terms, that CoCo's conversion and the issuer's
ulterior default are triggered as soon as a \emph{fundamental process}
$(U_{t})_{t\geq0}$ crosses, respectively, the levels $\bar{c}$ and
$\underline{c}$ , \emph{i.e.},
\begin{equation}\label{eq: conv and def}
\tau=\inf\{t\geq0\;:\; U_{t}\leq\bar{c}\}\qquad\text{and}\qquad\delta=\inf\{t\geq0\;:\; U_{t}\leq\underline{c}\}.
\end{equation}

The subsequent difference between models relies on how this fundamental
process is defined in order to serve as proxy for the regulatory capital,
and how the correspondent prices are evaluated. A prominent particular
choice for $(U_{t})_{t\geq0}$ is the given by
\begin{equation}
U_{t}:=\log\frac{S_{t}}{\ell_{t}},\qquad t\geq0,\label{eq: fundamental process}
\end{equation}
where $(S_{t})_{t\geq0}$ stands for the issuer's share price, and
$(\ell_{t})_{t\geq0}$ models the main benchmark for the issuer performance.
The importance of this process has been already addressed in earlier
literature on credit risk as in Longstaff and Schwarz \cite{LongstaffSchwartz}
and Sa\'a-Raquejo and Santa-Clara \cite{SaaRaquejoSantaClara}. The
process in (\ref{eq: fundamental process}) is also referred to as
\emph{log-leverage }or \emph{solvency ratio} process. We use the term
\emph{fundamental} since all credit events are assumed to be triggered
by the movements of $(U_{t})_{t\geq0}$ to a series of critical barriers.
Furthermore, the correspondent pricing formulas will be obtained
in terms of $(U_{t})_{t\geq0}$, not in terms of $(S_{t})_{t\geq0}$
and $(\ell_{t})_{t\geq0}$ independently.

Traditionally, structural models operate under two important but arguable
assumptions:
\begin{description}
\item [{\emph{$(\mathbf{A}_{1})$}}] The correlation between the noises
driving the share price and $(U_{t})_{t\geq0}$ is constant.
\item [{\emph{$(\mathbf{A}_{2})$}}] The fundamental process is fully observable
in continuous time.
\end{description}
However, in practice, it seems more reasonable to assume that this
is not the case, and assume instead the presence of a\emph{ short-term
uncertainty} in the following sense:
\begin{description}
\item [{\emph{$(\mathbf{A}'_{1})$}}] The correlation between the noises
driving the share price and $(U_{t})_{t\geq0}$ may vary.
\item [{\emph{$(\mathbf{A}'_{2})$}}] The fundamental process is fully
observable only at predetermined dates $(T_{j})_{t\in\mathbb{Z}_{+}}$.
\end{description}
Notice that, in some sense, the assumption $(\mathbf{A}'_{2})$ accounts
for the fact that in most cases regulatory capital depends on the
balance sheets of the firm issuing the CoCo, and those sheets are
updated only at series of predetermined dates $(T_{j})_{t\in\mathbb{Z}_{+}}$.
On the other hand, it seems more natural to assume $(\mathbf{A}'_{1})$
and include the correlation $\rho$ between the share price and the
fundamental process as another rightful model parameter. From now on we shall exclude the case where we have a perfect correlation between the stock and the fundamental process, \emph{i.e.}, the case where $\rho=1$.

Working under the short-term uncertainty has two immediate consequences.
On the one hand, and deviating from other structural models, the assumption
$(\mathbf{A}'_{1})$ prevents default and conversion times from being
stopping times with respect to the filtration generated by the relevant
state variables (\emph{e.g.}, share price, interest rates, total assets
value,...). On the other hand, assumption $(\mathbf{A}'_{2})$ implies
an information structure which is different from other partial information
models as Coculescu et al. \cite{CocGemJea}, Collin-Dufresne et al. \cite{CDufGolHel},
Duffie and Lando \cite{DuffieLando}, and Gou et al. \cite{GouJarZen}.
The work of Jeanblanc and Valchev \cite{JeaValc} is closer to our
model and, as mentioned in the introduction, this work combines techniques of \cite{DuffieLando} and \cite{JeaValc}, giving in some sense a two variate  extension of \cite{JeaValc}.
\section{Short-term uncertainty}\label{sec: short-term uncertainty}

\subsection{Introducing assumption ($\mathbf{A}_{1}'$).}

We shall assume that the $n$-dimensional $\mathbb{P}^{*}$-Brownian
motion $(W_{t}^{*}:=(W_{1,t}^{*},...,W_{n,t}^{*}))_{t\geq0}$ is the
noise driving all relevant state variables (\emph{e.g.}, share price,
interest rates, total assets value,...). In particular, we shall assume
that the $\mathbb{P}^{*}$-dynamics of the share price $(S_{t})_{t\geq0}$
be given by
\begin{equation}
\frac{\mathrm{d}S_{t}}{S_{t}}=r_{t}\mathrm{d}t+\sigma_{t}\cdot\mathrm{d}W_{t}^{*},\label{eq: share price}
\end{equation}
where we have used $(r_{t})_{t\geq0}$ to denote the interest rate,
and the $(\sigma_{t})_{t\geq0}$ is positive $\mathbb{R}^{n}$-valued
c\`adl\`ag process, integrable with respect to $(W_{t}^{*})_{t\geq0}$.
Hereafter the symbol $\cdot$ denotes the dot product in $\mathbb{R}^{n}$.

Additionally, given a correlation function $\rho:\mathbb{R}_{+}\rightarrow[-1,1]$,
let $(W_{t}^{\rho}:=(W_{j,t}^{\rho},...,W_{n,t}^{\rho}))_{t\geq0}$
be a second $\mathbb{P}^{*}$-Brownian motion satisfying $\mathrm{d}W_{j,t}^{*}\mathrm{d}W_{j,t}^{\rho}=\rho_{j}(t)\mathrm{d}t$.
In these terms, we consider the parametric family
\begin{equation}
\mathrm{d}U_{t}=\mu_{t}\mathrm{d}t+\sigma'_{t}\cdot\mathrm{d}W_{t}^{\rho},\label{eq: parametric U}
\end{equation}
where $(\sigma'_{t})_{t\geq0}$ is positive $\mathbb{R}^{n}$-valued
c\`adl\`ag process, integrable with respect to $(W_{t}^{\rho})_{t\geq0}$.
Let us note that at this point we don't assume any specific model
for volatilities $(\sigma_{t})_{t\geq0}$ and $(\sigma'_{t})_{t\geq0}$.

We shall further assume that all credit events are triggered by the
movements of $(U_{t})_{t\geq0}$ to a series of critical barriers.
More specifically, the conversion and default times will be given
respectively by
\begin{equation}
\tau=\inf\{t\geq0\;:\; U_{t}\leq\bar{c}\}\qquad\text{and}\qquad\delta=\inf\{t\geq0\;:\; U_{t}\leq\underline{c}\},\label{eq: parameteric tau and delta}
\end{equation}
For simplicity in the exposition we shall assume that no coupons are attached to the CoCo. However, we remark that the coupon cancellation feature introduced in \cite{Cocacoco}
can be included in a natural way by establishing that the $j$-th coupon will be paid at
time $T_{c_j}$ if and only if the $\tau_{j}(\rho)>T_{c_j}$ where
\begin{equation}
\tau_{j}:=\inf\{t\geq0\;:\; U_{t}\leq\bar{c}_{j}\},\qquad j=1,...,m.\label{eq: parameteric tau_j}
\end{equation}
It will be assumed that constants are ordered as $\underline{c}<\bar{c}\leq\bar{c}_{m}\leq...\leq\bar{c}_{1}$
so that the coupon cancellation, conversion and default may only occur
in a sequenced way.

\subsection{Introducing assumption ($\mathbf{A}_{2}'$).}

Let the set $\{T_{j},\;j\in\mathbb{Z}_{+}\}$ denote the times at which we fully observe $(W_{t}^{\rho})_{t\geq0}$. In addition, we assume that $(W_{t}^{*})_{t\geq0}$ is observed at times $\{t_{i,j}, \; i=0,...,n,\;j\in\mathbb{Z}_{+}\}$ where we set $T_0:=0$ and $t_{i,0}=T_i$. Let us define
\[
\left\lfloor t\right\rfloor :=\min\{T_{j}\in\{T_{0},T_{1},T_{2},...\}\;:\; T_{j}\leq t<T_{j+1}\},\qquad t\geq T_{0}.
\]
Let $\mathbb{F}^{W^{\rho}}=(\mathcal{F}_{t}^{W^{\rho}})_{t\geq0}$ and $\mathbb{F}^{W^{*}}=(\mathcal{F}_{t}^{W^{*}})_{t\geq0}$ denote
the natural filtration generated by $(W_{t}^{\rho})_{t\geq0}$ and $(W_{t}^{*})_{t\geq0}$, respectively. In these terms,
our reference information about the two noises $(W_{t}^{*})_{t\geq0}$ and $(W_{t}^{\rho})_{t\geq0}$ is given by $\widetilde{\mathbb{F}}=(\widetilde{F}_t)_{t\geq0}$ with
\[
\widetilde{\mathcal{F}}_{t}:=\mathcal{F}_{\left\lfloor t\right\rfloor}^{W^{\rho}}\vee\sigma(\{W^{*}_{t_{ij}},t_{ij} \leq t\}),\qquad t\geq0.
\]
Notice that $\mathcal{F}_{\left\lfloor t\right\rfloor}^{W^{\rho}}=\mathcal{F}_{t}^{W^{\rho}}$
for every $t\in\{T_{j}, \; j\in\mathbb{Z}_{+}\}$, but $\mathcal{F}_{t}^{W^{\rho}}\varsupsetneq\mathcal{F}_{\left\lfloor t\right\rfloor}^{W^{\rho}}$
otherwise. Thus in our model there is a noise which clears out at
update times $(T_{j})_{j\in\mathbb{Z}_{+}}$. Further, in between two update times, say $T_j$ and $T_{j+1}$, the correlated process $(S_t)_{t\geq0}$ provides a noisy information about $(U_t)_{t\geq0}$ by means of the observations $\{S_{t_{j,0}},S_{t_{j,1}}...,S_{t_{j,n}}\}$, where $T_j=t_{j,0}\leq...\leq t_{j,n}=T_{j+1}$.

As already pointed out in the introduction, there are some differences
between our model and other related approaches dealing with incomplete
information in the credit risk literature. For instance, Coculescu
et al. \cite{CocGemJea} consider a similar setting in which instead
of $(U_{t})_{t\geq0}$ they only observe a correlated process $(U_{t}^{\rho})_{t\geq0}$.
However, their observation of $(U_{t}^{\rho})_{t\geq0}$ is continuous,
whereas we only fully observe it at times $(T_{j})_{t\in\mathbb{N}}$.
Similar arguments apply for \cite{CDufGolHel} and Duffie and Lando
\cite{DuffieLando}. On the other hand, our model also differs from
that in Gou et al. \cite{GouJarZen} since the proposed update at times
disrupts the permanent delay in the arrival of information considered in  \cite{GouJarZen}.

On the other hand, our model also deviates from traditional structural
models since the credit events in (\ref{eq: parameteric tau and delta})
and (\ref{eq: parameteric tau_j}) are no longer stopping times with
respect to the reference filtration $\widetilde{\mathbb{F}}$. We shall consider the full market information $\mathbb{G}=(\mathcal{G}_t)_{t\geq 0}$ as given by $\mathcal{G}_t:=\widetilde{\mathcal{F}}_t\vee\mathcal{H}_t$ where $\mathcal{H}_t := \sigma(\textbf{1}_{\{\tau \leq s\}},\; s\leq t)$.
\begin{rem}
The short-term uncertainty model presented here can be readily stated
in a more general way, for instance by replacing $(W_{t}^{*})_{t\geq0}$
by a more general L\'evy process. Since here we are interested to obtain  analytical price formulas, we shall not consider
the possibility of including jumps in the share price (nor the fundamental
process) dynamics. However, we note here that a numerical study in
the presence of jumps can be conducted, for instance, using the techniques
in Metwally and Atiya\cite{MetAti}, Ruf and Scherer \cite{RufScherer}
or Hieber and Scherer \cite{HieScherer}.
\end{rem}

\subsection{Compensator of $\tau$}
Define the process $(F(t):=\mathbb{P}^*(\tau\leq t|\widetilde{\mathcal{F}}_{t}))_{t\geq0}$. In a standard structural model, $\tau$ would be a stopping time with respect to the reference filtration $\widetilde{\mathbb{F}}$, leading to the identity $F(t)=\textbf{1}_{\{\tau\leq t\}}$. Under the Assumption $\mathbf{A}_{2}'$, however, $\textbf{1}_{\{\tau\leq t\}}$ no longer belongs to $\widetilde{\mathcal{F}}_t$. Instead we can find the compensator of $\tau$ with respect to  $\widetilde{\mathbb{F}}$. Indeed, notice that $t\leq s$ implies $\{\tau\leq t\}\subseteq\{\tau\leq s\}$,
and
\[
\mathbb{E}^*[F(s)|\widetilde{\mathcal{F}}_{t}]=\mathbb{E}^*\left[\mathbb{P}(\tau\leq s|\widetilde{\mathcal{F}}_{s})\Big|\widetilde{\mathcal{F}}_{t}\right]=\mathbb{P}^*(\tau\leq s|\widetilde{\mathcal{F}}_{t})\geq\mathbb{P}^*(\tau\leq t|\widetilde{\mathcal{F}}_{t})=F(t).
\]
Consequently, $(F_{t})_{t\geq0}$ is an $\widetilde{\mathbb{F}}$-submartingale
and thus it admits a Doob-Meyer decomposition, \emph{i.e.}, it can
be written as $F(t)=M_{t}+A_{t}$ where $(M_{t})_{t\geq0}$ is an
$\widetilde{\mathbb{F}}$-martingale and $(A_{t})_{t\geq0}$ is a
$\widetilde{\mathbb{F}}$-predictable increasing process. When the
compensator $(A_{t})_{t\geq0}$ is known, we can proceed further and
say that process
\[
\textbf{1}_{\{\tau\leq t\}}-\int_{0}^{t\wedge\tau}\frac{\mathrm{d}A_{s}}{1-F(s^{-})},\qquad t\geq0,
\]
follows an $\widetilde{\mathbb{F}}$-submartingale. It turns out that,
similarly to \cite[Lemma 1]{JeaValc}, we have the following.

\begin{lem}\label{pro: compensator}The compensator of $(F(t))_{t\geq0}$ is given by
\[
A_{t}:=F(t)-\sum_{j:\; T_{j}\leq\left\lfloor t\right\rfloor }\Delta F(T_{j}),\qquad t\geq0.
\]
\end{lem}

\begin{proof}See Appendix \ref{sec: compensator}. \end{proof}

\section{Pricing CoCos under short-term uncertainty}
\subsection{The model for the fundamental process}

Let the $\mathbb{P}^{*}$-dynamics of $(B(t,T_{j}))_{t\geq0}$,\emph{
i.e.,} the price of a default-free bond with maturity $T$, are
given by
\begin{equation}
\mathrm{d}B(t,T)=B(t,T)\left(r_{t}\mathrm{d}t+b(t,T)\cdot\mathrm{d}W_{t}^{*}\right),\label{eq: PI general bond dynamics}
\end{equation}
where $(b(t,T))_{t\geq0}$ is an $\mathbb{F}^{W^{*}}$-adapted
process. As benchmark for the issuer performance we consider process
\begin{equation}
\ell_{t}:=LB(t,T)\exp\left\{ \int_{t}^{T}[\kappa(s)+a\left\Vert \sigma_{s}\right\Vert ^{2}]\mathrm{d}s\right\} ,\qquad t\geq0,\; a\in\mathbb{R}.\label{eq: PI general barrier}
\end{equation}

Notice that, by applying the It\^o formula, the log-leverage process
$(\log(S_{t}/\ell_{t}))_{t\geq0}$ satisfies
\begin{equation}
\mathrm{d}\log\frac{S_{t}}{\ell_{t}}=\left(a\left\Vert \sigma_{t}\right\Vert ^{2}-\frac{1}{2}\left\Vert \sigma_{t}\right\Vert ^{2}+\frac{1}{2}\left\Vert b(t,T)\right\Vert ^{2}\right)\mathrm{d}t+(\sigma_{t}-b(t,T))\cdot\mathrm{d}W_{t}^{*}.\label{eq: PI log-leverage process}
\end{equation}
Then we shall focus on the parametric family of fundamental processes given by
\begin{equation}
\mathrm{d}U_{t}:=\mathrm{d}\log\frac{S_{t}}{\ell_{t}}=\left(a\left\Vert \sigma_{t}\right\Vert ^{2}-\frac{1}{2}\left\Vert \sigma_{t}\right\Vert ^{2}+\frac{1}{2}\left\Vert b(t,T)\right\Vert ^{2}\right)\mathrm{d}t+(\sigma_{t}-b(t,T))\cdot\mathrm{d}W_{t}^{\rho},\label{eq: PI the fundamental process}
\end{equation}
where $\rho$ is our given deterministic correlation function.

It is important to remark that the barrier in (\ref{eq: PI general barrier})
conveys the two-folded appeal found in structural models in credit
risk: model features can be endowed with an economic interpretation,
and neat closed-form price formulas can be obtained in many cases.
Indeed, as argued by Bryis and De Varenne \cite{BryisDeVarenne},
$(\ell_{t})_{t\geq0}$ has many advantages when seen as a safety covenant.
In particular, regardless of the reorganization form, this barrier
allows to  define the bondholder's payoff upon default by relating
these payoffs to the level of the barrier. In this sense, if at time
$T$ the liabilities of the firm amount to the quantity $N$, then
our choice of the value $0\leq L\leq N$, models how protective this
safety covenant is. On the other hand, the factor $LB(t,T)$ exhibits
the barrier in (\ref{eq: PI general barrier}) as a direct extension
of the Merton \cite{Merton} and Black and Cox \cite{BlackCox} models
to a non-flat barrier with stochastic interest rates. Moreover, the
factor $\exp\{a\int_{t}^{T}\sigma^{2}(s)\mathrm{d}s\}$ provides the
extra parameter $a$ which allows us modify the exponential profile
of the barrier, by increasing its concavity and steepness, see Brigo
and Tarenghi \cite{BrigoAT1P}. Altogether, $(\ell_{t})_{t\geq0}$
can accomodate other models from the earlier credit risk literature
as Longstaff and Schwarz \cite{LongstaffSchwartz} and Sa\'a-Raquejo
and Santa-Clara \cite{SaaRaquejoSantaClara}. From the computational
point of view, let us mention that Lo et al. \cite{LoLeeHui} and
Rapisarda \cite{Rapisarda} are able to compute accurate closed-form
estimates for time-dependent Black-Scholes option prices by applying
the so-called \emph{mirror image approach} in order to solve the PDE
associated with the price. It is apparent that the only parametric
form of barrier that is compatible with the mirror image fundamental
solution is that in (\ref{eq: PI general barrier}), for deterministic
$(r_{t})_{t\geq0}$. Remarkably, \cite{LoLeeHui} present upper and
lower bounds to the barrier option price by applying the maximum
principle for the diffusion equation, which translates into a simple
and intuitive argument on the barrier profile.

Since our focus is on the fundamental process, and not on the barrier,
we have decided to translate Assumption $(\mathbf{A}_{1}')$ as the
correlation structure between the noises driving $(S_{t})_{t\geq0}$
and $(U_{t})_{t\geq0}$ . Of course we could alternatively introduce
a new stochastic factor in the barrier $(\ell_{t})_{t\geq0}$ ---say,
replacing $L$ by $(L_{t})_{t\geq0}$--- but an economical interpretation
for this new factor should be discussed first. Moreover, as argued
in \cite{LongstaffSchwartz}, taking a more complex barrier has the
inherent risk of resulting in a more involved model providing no addition
insight. Besides, as we shall see below, the obtainment of analytical
formulas depends explicitly on the vector $(U_{t},\log S_{t})$,
not separately on $(S_{t})_{t\geq0}$ and $(\ell_{t})_{t\geq0}$.

Moreover, this model accommodates also recent contributions on the study of CoCos such as \cite{BrigoCoco,Cocacoco}. Indeed we have the following.
\begin{example} \label{exa: brigo}
The Brigo et al. \cite{BrigoCoco} model corresponds to a one-dimensional
case, where $\sigma$ is a piecewise constant deterministic function,
and $a$ (in their notation $a:=B$) is given as one of the model
parameters. The correlation is assumed to equal to $1$, so that
\begin{equation*}\label{exa: brigo u}
\mathrm{d}U_{t}=(a-\tfrac{1}{2})\sigma^{2}(t)\mathrm{d}t+\sigma(t)\mathrm{d}W_{t}^{*}.
\end{equation*}

\end{example}

\begin{example} \label{exa: cocacoco}
The Corcuera et al. \cite{Cocacoco} model corresponds to a $n$-dimensional
case, with a possibly stochastic volatility $(\sigma_{t})_{t\geq0}$
and bond prices given as in the Gaussian HJM framework so that $b$
is deterministic. In this case correlation is also to equal to $1$,
and the fundamental process is given
\begin{equation*} \label{exa: cocacoco u}
\mathrm{d}U_{t}=\tfrac{1}{2}(-\left\Vert \sigma_{t}\right\Vert ^{2}+\left\Vert b(t,T)\right\Vert ^{2})\mathrm{d}t+(\sigma_{t}-b(t,T))\cdot\mathrm{d}W_{t}^{*}.
\end{equation*}
\end{example}

\subsection{The pricing problem}\label{sec: pricing problem}

The general price formula for a CoCo in our setting can be directly
derived from \cite[Lemma 1]{Cocacoco}. More specifically, for every
$0\leq t\leq T$, the price of a CoCo, on $\{\tau>t\}$, is given
by
\begin{eqnarray}
\pi(t) & := &
NB(t,T)\mathbb{P}^{T}(\tau>T|\;\mathcal{G}_{t})+\frac{C_{r}S_{t}}{\mathrm{e}^{\int_{t}^{T}\kappa(u)\mathrm{d}u}}\mathbb{P}^{(S)}(\tau\leq T|\;\mathcal{G}_{t}),\label{eq: the price}
\end{eqnarray}
where $(B(t,T))_{t\geq0}$ stands for the price of a default-free
bond with maturity $T$, and the $T$-forward measure $\mathbb{P}^{T}$
(resp. share measure $\mathbb{P}^{(S)}$) is the probability measure
given by taking $(B(t,T))_{t\geq0}$ (resp. $(S_{t})_{t\geq0}$)
as \emph{num\'eraire}. That is to say, these probability measures are equivalent
to $\mathbb{P}^{*}$ and their Radon-Nikod\'{y}m derivatives are
given by
\[
\frac{\mathrm{d}\mathbb{P}^{T}}{\mathrm{d}\mathbb{P}^{*}}=\frac{\mathrm{e}^{-\int_{0}^{T}r_{u}\mathrm{d}u}B(T,T)}{B(0,T)}=\exp\left\{ \int_{0}^{T}b(u,T)\mathrm{d}W_{u}^{*}-\frac{1}{2}\int_{0}^{T}b^2(u,T)\mathrm{d}u\right\} ,
\]
and
\[
\frac{\mathrm{d}\mathbb{P}^{(S)}}{\mathrm{d}\mathbb{P}^{*}}=\frac{\mathrm{e}^{-\int_{0}^{T_{j}}[r_{u}-\kappa(u)]\mathrm{d}u}S_{T}}{S_{0}}=\exp\left\{ \int_{0}^{T}\sigma_{u}\mathrm{d}W_{u}^{*}-\frac{1}{2}\int_{0}^{T}\sigma_{u}^{2}\mathrm{d}u\right\} .
\]
respectively.

Notice that according to (\ref{eq: the price}), in other to price a CoCo we need to find expressions for the distribution of $\tau$ under both $\mathbb{P}^{T}$ and $\mathbb{P}^{(S)}$. Therefore the rest of this work is devoted to study the obtainment of the aforementioned distributions.

\subsection{The distribution of $\tau$}
We shall consider a concrete case by assuming that the interest rate, the  correlation function and the stock's volatility are strictly positive constants, \emph{i.e.}, $r_t\equiv r$ and $\rho(t)\equiv \rho$, and $\sigma_t \equiv \sigma$.
In light of Proposition (\ref{pro: PI general dynamics}) in the Appendix, the we have the following dynamics of (2) and (3) under $\mathbb{P}^T$

\begin{equation}\label{eq: system under P*}
\begin{cases}
\mathrm{d}\log S_t  =   \mu_S^{*} \mathrm{d}t + \sigma \mathrm{d}W^{*}_t\\
\mathrm{d}U_t  =   \mu_U^{*} \mathrm{d}t + \sigma \mathrm{d}W^{\rho}_t\\
\end{cases}
\end{equation}
where $\mu_S^{*}:=r-\tfrac{1}{2}\sigma^2$ and $\mu^*_U :=(a-\tfrac{1}{2})\sigma^2$. And similarly under $\mathbb{P}^{(S)}$ we have
\begin{equation}\label{eq: system under P(S)}
\begin{cases}
\mathrm{d}\log {S}_t  =   \mu_S^{(S)} \mathrm{d}t + \sigma \mathrm{d}B_t\\
\mathrm{d}U_t  =   \mu_U^{(S)} \mathrm{d}t + \sigma \mathrm{d}B^{\rho}_t\\
\end{cases}
\end{equation}
where $\mu_S^{(S)} :=r+\tfrac{1}{2}\sigma^2$, $\mu_U^{(S)} :=(a-\tfrac{1}{2}+\rho)\sigma^2$, and $(B_t)_{t\geq0}$ and $(B_t)_{t\geq0}$ are two $\mathbb{P}^{(S)}$-Brownian motions with correlation $\rho$.  Let us now focus on the first summand of (\ref{eq: the price}) since the second can be computed in an analogous manner. Moreover, let us consider the case where the CoCo maturity coincides with the first time the fundamental process is fully updated, that is to say, we set $T=T_1$. Then for $0\leq t \leq T$, by conditioning to $\mathcal{F}^{W^{\rho}}_t\vee\sigma(\{W^{*}_{t_{ij}},t_{ij} \leq t\})\supseteq\mathcal{G}_t$, we can use a known result on Brownian motions hitting times (see \cite{bor12}) in order to get
\begin{equation}\label{eq: the conditioning step}
\mathbb{P}^{*}(\tau >T|\mathcal{G}_{t})=\mathbf{1}_{\{\tau >t\}}\left( \mathbb{E%
}^{*}\left[ \left. \Phi \left( -d^*_{-}\right) -\exp\left\{-\frac{2\mu_U^*(U_{t}-\bar{c})}{%
\sigma^{2}}\right\}\Phi \left( d^*_{+}\right) \right\vert \mathcal{G}_{t}\right]
\right),
\end{equation}
where $\bar{c}$ is the lower threshold for conversion (as in (\ref{eq: conv and def})) and
\[
d^*_{\pm }:=\frac{\bar{c}-U_{t}\pm \mu^*_U(T-t)}{\sigma\sqrt{T-t}}.
\]
The problem is thus reduced to compute the density of $U_t$ conditioned to $\mathcal{G}_t$. Suppose that up to time $t$ we have observations of the stock at times  $0=t_{0,0}\leq t_{0,1} \leq ... \leq t_{0,k}=t$, then we look for the density
\[
\mathbb{P}^{*}\left( U_t\in \mathrm{d}u_t\Big\vert \tau>t,S_{t_{0,0}}\in \mathrm{d}s_0,S_{t_{0,1}}\in \mathrm{d}s_1,...,S_{t_{0,k}}\in \mathrm{d}s_k\right).
\]
For ease of notation, for a process $(X_t)_{t\geq0}$ let us write $X_{j}:=X_{t_{0,j}}$, $X^{(k)}=(X_{0},X_{1},...,X_{k})$, $x^{(k)}=(x_{0},x_{1},...,x_{k})$. In these terms and using the Bayes rule we can rewrite the density above as
\begin{eqnarray*}
\mathbb{P}^{*}\left( U_t\in \mathrm{d}u_t\Big\vert \tau
>t,S^{(k)}\in \mathrm{d}s^{(k)}\right) &=&\frac{\mathbb{P}\left( \tau
>t,U_t\in \mathrm{d}u_t,S^{(k)}\in \mathrm{d}s^{(k)}\right) }{%
\mathbb{P}^{*}\left( \tau >t,S^{(k)}\in \mathrm{d}s^{(k)}\right) } \\
& = & \frac{\int_{\mathbb{R}^{k-1}}\mathbb{P}\left( \tau
>t,U^{(k)}\in \mathrm{d}u^{(k)},S^{(k)}\in \mathrm{d}s^{(k)}\right)
\mathrm{d}u^{(k-1)}}{\int_{\mathbb{R}^{k}}\mathbb{P}\left( \tau
>t,U^{(k)}\in \mathrm{d}u^{(k)},S^{(k)}\in \mathrm{d}s^{(k)}\right)
\mathrm{d}u^{(k)}}.
\end{eqnarray*}%
Using Theorem \ref{the result} in the Appendix (with $X:=U$, $Y:=\rho\log S$ and $Z:=X-Y$) we can show that actually
\begin{eqnarray*}
\lefteqn{\mathbb{P}^{*}\left(\tau>t,U_{k}^{(k)}\in\mathrm{d}u^{(k)},S^{(k)}\in\mathrm{d}s^{(k)}\right)} \\
 & = & \prod\limits _{j=1}^{k}\left(\left(1-\exp\left\{ -\frac{2(u_{j}-\bar{c})(u_{j-1}-\bar{c})}{\sigma^{2}(t_{0,j}-t_{0,j-1})}\right\}\right)\textbf{1}_{\{u_j>\bar{c},u_{j-1}>\bar{c}\}}\right) h^{*}_j \big(u_j - u_{j-1}-\rho\log \tfrac{s_{j}}{s_{j-1}}\big)g^{*}_j\big(\rho\log\tfrac{s_{j}}{s_{j-1}}\big)
\end{eqnarray*}
where $h^*_j$ and $g^*_j$   are the following Gaussian densities
\[
h^{*}_{j}(z):=\frac{1}{\sigma\sqrt{1-\rho^{2}}\sqrt{2\pi(t_{0,j}-t_{0,j-1})}}\exp\left\{ -\frac{\left(z-(\mu_U^{*}-\rho\mu_S^{*})(t_{0,j}-t_{0,j-1})\right)^{2}}{2\sigma^{2}(1-\rho^{2})(t_{0,j}-t_{0,j-1})}\right\} .
\]
\[
g^{*}_{j}(x):=\frac{1}{\sigma\rho\sqrt{2\pi(t_{0,j}-t_{0,j-1})}}\exp\left\{ -\frac{\left(x-\rho\mu_S^{*}(t_{0,j}-t_{0,j-1})\right)^{2}}{2\rho^{2}\sigma^{2}(t_{0,j}-t_{0,j-1})}\right\},
\]
In light (\ref{eq: system under P*}) and (\ref{eq: system under P(S)}), it is clear that the computation of $\mathbb{P}^{(S)}(\tau>T|\mathcal{G}_t)$ can be conducted in the same fashion, and thus we obtain an expression for the price in (\ref{eq: the price}).

It is worth noticing that, within each interval $[T_{j},T_{j+1})$,  the survival probabilities obtained by Jeanblanc
and Valchev \cite{JeaValc} depend only on the past value $S_{T_j}$. In our setting, however, these probabilities depend on the series $S_{t_{j,0}}$,...,$S_{t_{j,k}}$ where $T_j=t_{j,0}\leq t_{j,1} \leq ... \leq t_{j,k} \leq t$. This
difference relies on the fact that even though within each interval
$[T_{j},T_{j+1})$, our knowledge of the fundamental process is constant, we still observe the evolution
of all the other ($\mathbb{F}^{W^{*}}$-adapted) state variables at selected times.

If we compare the price in this model with that in a model where the fundamental process is observed, intuitively, we expect the following: the stock  is a proxy for the fundamental value in such a way that $(U_t)_{t \geq 0}$ moves around $(S_t)_{t\geq 0}$ and the price is an average, but the conversation time $\tau$ also says something about the behavior of $(U_t)_{t \geq 0}$, if  $\tau>t$ it means that $(U_t)_{t \geq 0}$ behaved better than $(S_t)_{t\geq 0}$,  especially if $(S_t)_{t\geq 0}$ is  low. So, in general,  we will get higher prices than in the models with total information and this will be more evident if $(S_t)_{t\geq 0}$ is low and/or its  correlation  with $(U_t)_{t \geq 0}$  is also low.

\begin{rem} For a discussion on how the pricing problem is changed by extending this \emph{base case} (\emph{i.e.}, with $r_t\equiv r$, $\rho(t)\equiv \rho$ and $\sigma_t \equiv \sigma$) to a model with stochastic interest rates or volatility and a time-varying correlation, we refer to \cite{Cocacoco}. More specifically,  \cite[Sections 4 and 5]{Cocacoco}  could be used to cover the case of stochastic interest rates and volatility, whereas \cite[Lemma 4]{Cocacoco} could provide an approximate formula when the correlation is time-varying. The main issue with these extensions relays on the fact that the expression within parentheses in (\ref{eq: the conditioning step}) cannot be obtained in closed-form as in the base case. Notice however that the base case discussed here can easily be extended a piecewise function, taking constant values within each interval $[T_{j},T_{j+1})$.
\end{rem}

\subsection{Numerical illustration}

For this part we fix the following parameters $r =  0.03$, $\sigma  =  .49$, $S_{0}  =  100$, $N  =  100$, $\bar{c}  = \log(35)$, $\kappa  =  0$, $\mu_{S}^{*}  =  r-\kappa-\frac{1}{2}$, and $\mu_{U}^{*}  =  \mu_{S}^{*}$, along with the four scenarios $\{\omega_{i}\}_{i=1,...,4}$ for the stock price depicted in the Figure 1. Then we shall estimate the effect
of the parameter $\rho$ on the probability $\mathbb{P}^{*}(\tau>t|\mathcal{G}_{t})$. We consider the following cases $\rho\in\{0.01,0.25,0.5,0.75,0.99\}$.
\begin{figure}[H]
\begin{centering}
\includegraphics[width=1\textwidth]{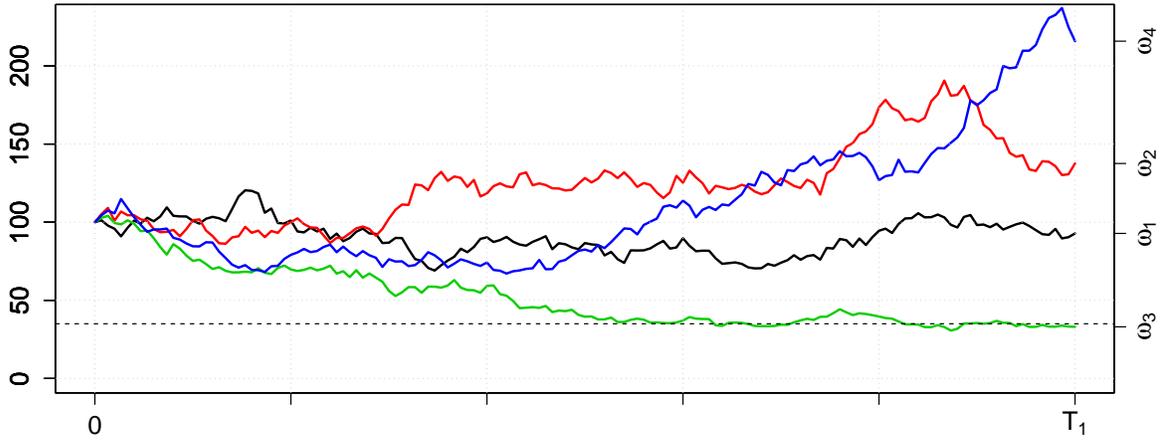}
\par\end{centering}

\protect\caption{Simulated trajectories of the stock price. The dashed line is set at level $\exp\{{\bar{c}\}}$.}
\end{figure}

Based on the scenario $(S_{t}(\omega_1))_{t\geq0}$ and different values for $\rho$, Figure 2 shows a series of simulated trajectories for the fundamental process $(U_t)_{t\geq 0}$. As the correlation parameter tends to 1, the behavior of   $(U_t)_{t\geq 0}$ matches that of the $(\log S_t)_{t\geq 0}$ (\emph{cf.} Example 4.1 and Example 4.2).

\begin{figure}[H]
\begin{centering}
\includegraphics[width=1\textwidth]{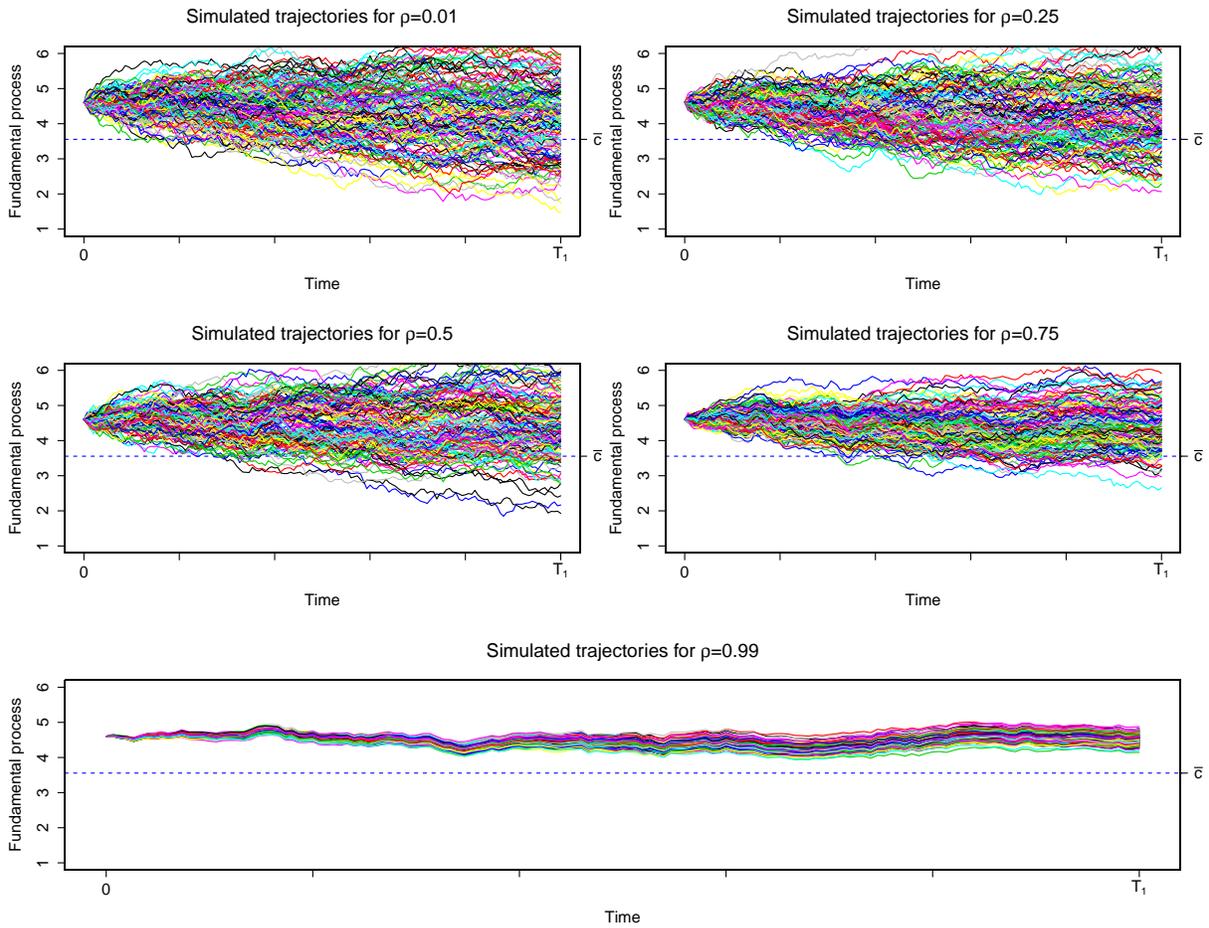}
\par\end{centering}

\protect\caption{Simulated trajectories of the fundamental process ased on the scenario $(S_{t}(\omega_1))_{t\geq0}$. The dashed line is set at level $\bar{c}$.}
\end{figure}

Finally Figure 3 shows the effect of the correlation parameter on the probability of avoiding conversion before $T_1$, that is, $\mathbb{P}^*(\tau>T_1 |\mathcal{G}_t)$.

\begin{figure}[H]
\begin{centering}
\includegraphics[width=1\textwidth,keepaspectratio]{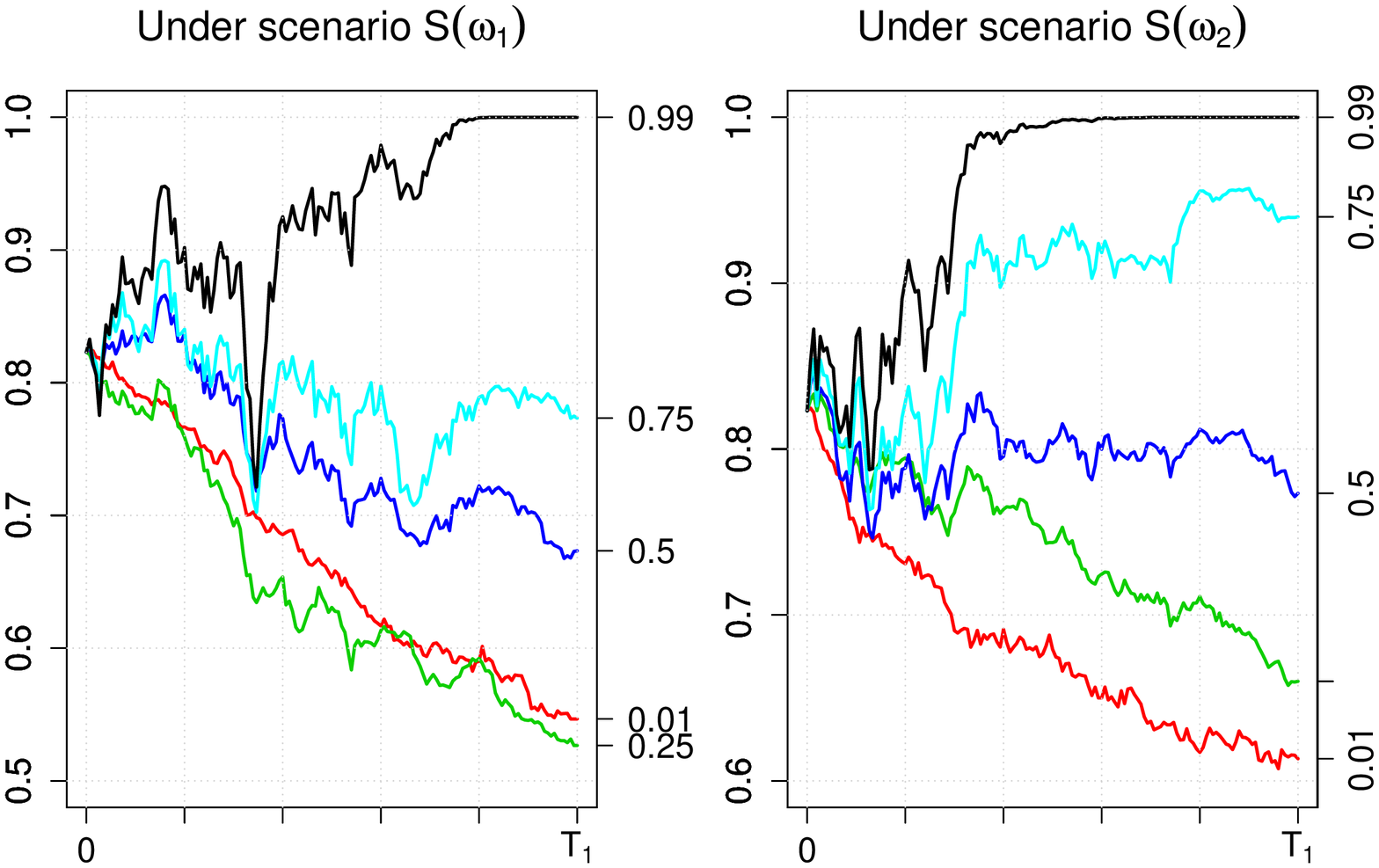}
\includegraphics[width=1\textwidth,keepaspectratio]{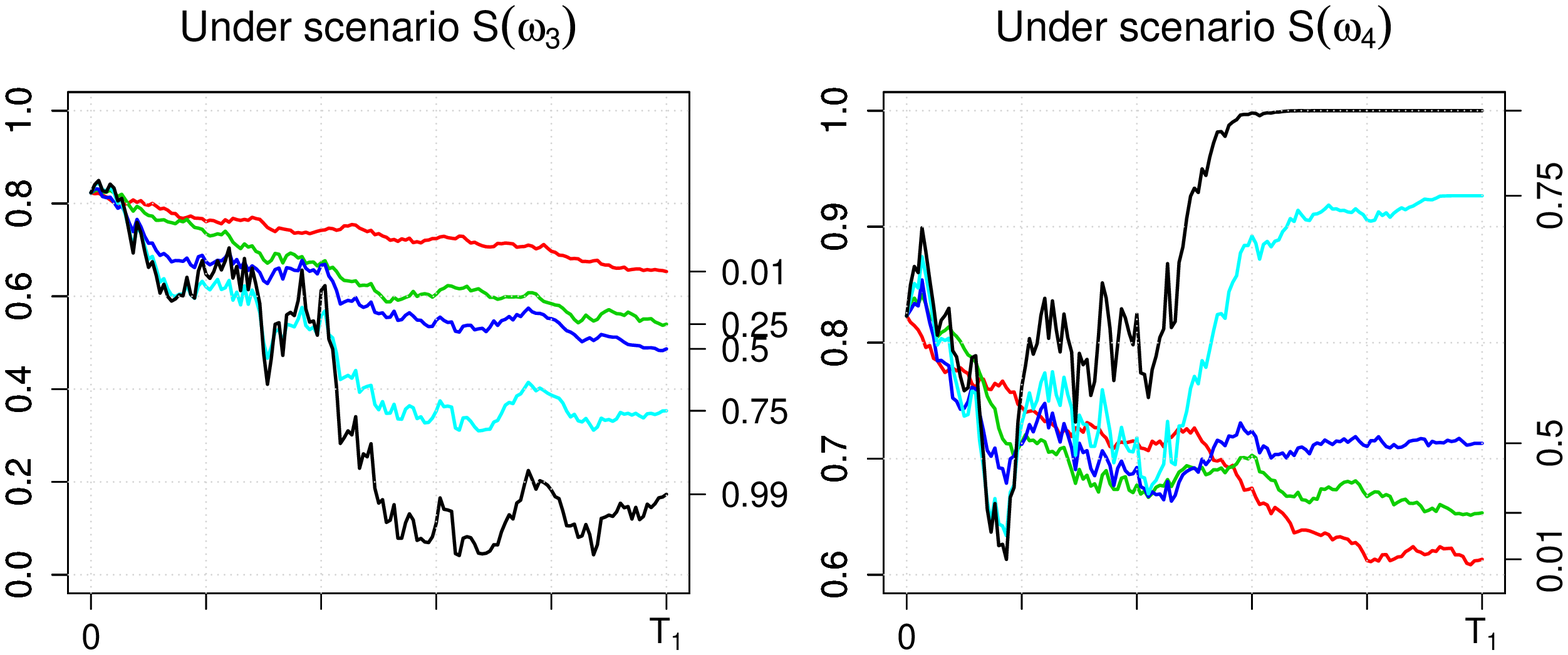}
\par\end{centering}

\protect\caption{Simulated values for $\mathbb{P}^*(\tau>T_1 |\mathcal{G}_t)$ based on different scenarios for the stock and correlation values (marks on the seconday axis).}
\end{figure}

\section*{Acknowledgements}
This work started while visiting the \textbf{Centre for Advanced Study} (CAS) at the Norwegian Academy of Science and Letters.  The authors would like to thank the kind hospitality offered by the members of the CAS.

\appendix
\section{Appendix}

\subsection{Proof of Lemma \ref{pro: compensator}}\label{sec: compensator}
Let us begin by showing that the process defined by
$M_{t}:=\sum_{j:\; T_{j}\leq\left\lfloor t\right\rfloor }\Delta F(T_{j}),$
$t\geq0$, follows an $\widetilde{\mathbb{F}}$-martingale. To this
matter, notice that by construction $M_{t}=M_{T_{j}}=M_{s}$ whenever
$T_{j}\leq s<t\leq T_{j+1}$, and thus $\mathbb{E}^*[M_{t}|\widetilde{F}_{s}]=M_{s}$.
Now if instead we have $T_{j}\leq s<T_{j+1}\leq t<T_{j+2}$, then
$M_{t}=M_{ T_{j}}+\Delta F(T_{j+1})$. Again we have $M_{T_{j}}=M_{s}$.
Moreover we have
\begin{eqnarray*}
\mathbb{E}^*[\Delta F(T_{j+1})|\widetilde{F}_{s}] & = & \mathbb{E}^*\left[\mathbb{P}^*(\tau\leq T_{j+1}|\widetilde{\mathcal{F}}_{T_{j+1}})\Big|\widetilde{\mathcal{F}}_{s}\right]-\mathbb{E}\left[\mathbb{P}^*(\tau\leq T_{j+1}^{-}|\widetilde{\mathcal{F}}_{T_{j+1}^{-}})\Big|\widetilde{\mathcal{F}}_{s}\right]\\
 & = & \mathbb{E}^*\left[\mathbf{1}_{\{\tau\leq T_{j+1}\}}\Big|\widetilde{\mathcal{F}}_{s}\right]-\mathbb{E}^*\left[\mathbf{1}_{\{\tau\leq T_{j+1}^{-}\}}\Big|\widetilde{\mathcal{F}}_{s}\right]\\
 & = & \mathbb{E}^*\left[\mathbf{1}_{\{\tau\leq T_{j+1}\}}-\mathbf{1}_{\{\tau<T_{j+1}\}}\Big|\widetilde{\mathcal{F}}_{s}\right]\\
 & = & 0,
\end{eqnarray*}
where the last equation follows from the fact $\mathbf{1}_{\{\tau\leq T_{j+1}\}}=\mathbf{1}_{\{\tau<T_{j+1}\}}$
since $\tau$ is defined as the hitting time of a Brownian motion.

By construction, $(A_{t})_{t\geq0}$ is a continuous
process, thus it is predictable. In order to see that it is increasing
we proceed by contradiction: On the one hand, suppose that there exists
$t< s$ such that $A_{t}\geq A_{s}$ \emph{a.s.;} this would imply
that $\mathbb{E}^*[A_{t}]\geq\mathbb{E}^*[A_{s}]$. On the other hand,
since $(F_{t})_{t\geq0}$ is an $\widetilde{\mathbb{F}}$-submartingale
we have
\[
\mathbb{E}^*[A_{s}]=\mathbb{E}^*[F(s)]-\mathbb{E}^*[M_{s}]\geq\mathbb{E}^*[F(t)]-\mathbb{E}^*[M_{s}]=\mathbb{E}^*[F(t)]-\mathbb{E}^*[M_{t}]=\mathbb{E}^*[A_{t}],
\]
and thus the equation $\mathbb{E}^*[A_{s}]=\mathbb{E}^*[A_{t}]$ should
hold. However, this would imply that $\mathbb{E}^*[F_{s}]=\mathbb{E}^*[F_{t}]$,
or equivalently $\mathbb{P}^*(\tau\leq s)=\mathbb{P}^*(\tau\leq t)$,
which is impossible since $\tau$ follows a strictly increasing distribution
---namely the Inverse-Gaussian distribution.

\subsection{The fundamental process dynamics under $\mathbb{P}^{T}$ and $\mathbb{P}^{(S)}$}
The following result describes the dynamics of the processes involved in the pricing problem, both under the $T$-forward measure  and share measure $\mathbb{P}^{(S)}$.

\begin{prop}
\label{pro: PI general dynamics}Let $(Z_{t})_{t\geq0}$ be a second Brownian motion, independent of $(W_{t}^{*})_{t\geq0}$  so that we can write
\[
W_{t}^{\rho}=\rho(t)W_{t}^{*}+\sqrt{1-\rho^{2}(t)}Z_{t},\qquad t\geq0.
\]
Let $\mathbb{F}^{W*}=(\mathcal{F}_{t}^{W*})_{t\geq0}$ and let $\mathbb{F}^{Z}=(\mathcal{F}_{t}^{Z})_{t\geq0}$ be the natural filtrations of  $W*$ and $Z$ respectively. Let $(U_{t})_{t\geq0}$ be
the process defined in (\ref{eq: PI the fundamental process}). Then (i) the process $(W_{t}^{T}:=W_{t}^{*}-\int_{0}^{t}b(t,T)\mathrm{dt})_{t\geq0}$
(resp. $(W_{t}^{(S)}:=W_{t}^{*}-\int_{0}^{t}\sigma_{t}\mathrm{dt})_{t\geq0}$)
is an $\mathbb{F}^{W^{*}}$-Brownian motion under $\mathbb{P}^{T}$
(resp. $\mathbb{P}^{(S)}$); $(Z_{t})_{t\geq0}$ is an $\mathbb{F}^{Z}$-Brownian
motion both with respect to $\mathbb{P}^{T}$ and $\mathbb{P}^{(S)}$;
and
\[
(\rho(t)W_{t}^{T}+\sqrt{1-\rho^{2}(t)}Z_{t})_{t\geq0}\qquad\left(resp.\quad(\rho(t)W_{t}^{(S)}+\sqrt{1-\rho^{2}(t)}Z_{t})_{t\geq0}\right),
\]
is a $\mathbb{G}$-Brownian motion under $\mathbb{P}^{T}$ (resp.
$\mathbb{P}^{(S)}$).

(ii) The dynamics of $(U_{t})_{t\geq0}$ under $\mathbb{P}^{T}$
are
\begin{eqnarray}
\mathrm{d}U_{t} & = & \left(a\left\Vert \sigma_{t}\right\Vert ^{2}-\frac{1}{2}\left\Vert \sigma_{t}\right\Vert ^{2}+\frac{1}{2}\left\Vert b(t,T)\right\Vert ^{2}+\rho(t)(\sigma_{t}-b(t,T))\cdot b(t,T)\right)\mathrm{d}t\nonumber \\
 &  & +(\sigma_{t}-b(t,T))\cdot\left(\rho(t)\mathrm{d}W_{t}^{T}+\sqrt{1-\rho^{2}(t)}\mathrm{d}Z_{t}\right).\label{eq: PI general PT dynamics}
\end{eqnarray}

(iii) The dynamics of $(U_{t})_{t\geq0}$ under $\mathbb{P}^{(S)}$
are
\begin{eqnarray}
\mathrm{d}U_{t}& = & \left(a\left\Vert \sigma_{t}\right\Vert ^{2}-\frac{1}{2}\left\Vert \sigma_{t}\right\Vert ^{2}+\frac{1}{2}\left\Vert b(t,T)\right\Vert ^{2}+\rho(t)(\sigma_{t}-b(t,T))\cdot\sigma_{t}\right)\mathrm{d}t\nonumber \\
 &  & +(\sigma_{t}-b(t,T))\cdot\left(\rho(t)\mathrm{d}W_{t}^{(S)}+\sqrt{1-\rho^{2}(t)}\mathrm{d}Z_{t}\right).\label{eq: PI general P(S) dynamics}
\end{eqnarray}
\end{prop}
\begin{proof}
\textit{\emph{The statement about }}$(W_{t}^{T})_{t\geq0}$ and
$(W_{t}^{(S)})_{t\geq0}$\textit{\emph{ follows from the Girsanov
theorem. }}Due to the independence (under $\mathbb{P}^{*}$) between
$(W_{t}^{*})_{t\geq0}$ and $(Z_{t})_{t\geq0}$, it is easy to see
that the L\'evy characterization of the Brownian motion applies to $(Z_{t})_{t\geq0}$
and, subsequently, it also applies to $(\rho(t)W_{t}^{T}-\sqrt{1-\rho^{2}(t)}Z_{t})_{t\geq0}$
and $(\rho(t)W_{t}^{(S)}-\sqrt{1-\rho^{2}(t)}Z_{t})_{t\geq0}$. In
light of (\emph{i}), it only remains to rewrite the dynamics in (\ref{eq: PI log-leverage process})
using the newly defined Brownian motions.\end{proof}

Notice that taking a constant interest rate implies that $\mathbb{P}^*$ and $\mathbb{P}^T$ coincide.

\subsection{Conditional densities}
\begin{thm} \label{the result}
Let $(X_t)_{t\geq0}$, $(Y_t)_{t\geq0}$, $(Z_t)_{t\geq0}$ be three drifted Brownian motions under $\mathbb{P}$, with $(Z_t)_{t\geq0}$ independent of $(Y_t)_{t\geq0}$, and such that  $X_t=Y_t+Z_t$. Let $0=t_0\leq t_1 \leq ...\leq t_n=t$ be a partition of $[0,t]$. Set the notations $X_{j}:=X_{t_j}$, $X^{(n)}=(X_{0},X_{1},..,X_{n})$, $x^{(n)}=(x_{0},x_{1},..,x_{n})$, and analogously for $Y$ and $Z$. Set $\tau =\inf \{t,X_{t}\leq c\}$.
Then the following equation holds true
\begin{eqnarray*}
\lefteqn{\mathbb{P}\left(\tau>t_{n},X^{(n)}\in\mathrm{d}x^{(n)},Y^{(n)}\in\mathrm{d}y^{(n)}\right)} \\
 & = & \prod\limits _{k=1}^{n}\left(1-\exp\left\{ -\frac{2(x_{k}-c)(x_{k-1}-c)}{\sigma_{X}^{2}(t_{k}-t_{k-1})}\right\} \right)\\
 &  & \qquad\times f_{\Delta Z_{k}}(x_{k}-y_{k}-(x_{k-1}-y_{k-1)}))f_{\Delta Y_{k}}(y_{k}-y_{k-1})
\end{eqnarray*}
where $\sigma^2_X:=\mathbb{E}[(X_1-\mathbb{E}[X_1])^2]$,  and the functions $f_{\Delta Z_{k}}$ and $f_{\Delta Y_{k}}$ denote the density of $\Delta Z_{k}:=Z_{k}-Z_{k-1}$ and $\Delta Y_{k}:=Y_{k}-Y_{k-1}$, respectively.
\end{thm}
\begin{proof}
We have that
\begin{eqnarray}
\lefteqn{\mathbb{P}\left( \tau >t_{n},X_{n}^{(n)}\in \mathrm{d}x^{(n)},Y^{(n)}\in
\mathrm{d}y^{(n)}\right)}   \nonumber \\
&=&\mathbb{P}\left( \tau >t_{n},X_{n}\in \mathrm{d}x_{n},Y_{n}\in \mathrm{d}%
y_{n},X^{(n-1)}\in \mathrm{d}x^{(n-1)},Y^{(n-1)}\in \mathrm{d}%
y^{(n-1)}\right)   \nonumber \\
&=&\mathbb{P}\left( \tau >t_{n},X_{n}\in \mathrm{d}x_{n},Y_{n}\in
\mathrm{d}y_{n}\Big\vert \tau >t_{n-1},X^{(n-1)}\in \mathrm{d}%
x^{(n-1)},Y^{(n-1)}\in \mathrm{d}y^{(n-1)}\right)   \nonumber \\
&&\times \mathbb{P}\left( \tau >t_{n-1},X^{(n-1)}\in \mathrm{d}%
x^{(n-1)},Y^{(n-1)}\in \mathrm{d}y^{(n-1)}\right)   \nonumber \\
&=&\mathbb{P}\left( \tau >t_{n},X_{n}\in \mathrm{d}x_{n},Y_{n}\in
\mathrm{d}y_{n}\Big\vert \tau >t_{n-1},X_{n-1}\in \mathrm{d}%
x_{n-1},Y_{n-1}\in \mathrm{d}y_{n-1}\right)   \nonumber \\
&&\times \mathbb{P}\left( \tau >t_{n-1},X^{(n-1)}\in \mathrm{d}%
x^{(n-1)},Y^{(n-1)}\in \mathrm{d}y^{(n-1)}\right) ,  \label{joint1}
\end{eqnarray}%
where the last equality is due to the Markovianity of $\left( \inf_{0\leq
s\leq t}X_{s},X_{t},Y_{t}\right) _{t\geq 0}$. Now, by the Bayes rule
\begin{eqnarray*}
\lefteqn{\mathbb{P}\left( \tau >t_{n},X_{n}\in \mathrm{d}x_{n},Y_{n}\in
\mathrm{d}y_{n}\Big\vert \tau >t_{n-1},X_{n-1}\in \mathrm{d}%
x_{n-1},Y_{n-1}\in \mathrm{d}y_{n-1}\right)}  \\
&=&\frac{\mathbb{P}\left( \tau >t_{n},X_{n}\in \mathrm{d}x_{n},Y_{n}\in
\mathrm{d}y_{n},X_{n-1}\in \mathrm{d}x_{n-1},Y_{n-1}\in \mathrm{d}%
y_{n-1}\right) }{\mathbb{P}\left( \tau >t_{n-1},X_{n-1}\in \mathrm{d}%
x_{n-1},Y_{n}\in \mathrm{d}y_{n}\right) } \\
&=&\mathbb{P}\left( \tau >t_{n}\Big\vert\tau >t_{n-1},X_{n}\in \mathrm{d}%
x_{n},Y_{n}\in \mathrm{d}y_{n},X_{n-1}\in \mathrm{d}x_{n-1},Y_{n-1}\in
\mathrm{d}y_{n-1}\right)  \\
&&\times \frac{\mathbb{P}\left( \tau >t_{n-1},X_{n}\in \mathrm{d}%
x_{n},Y_{n}\in \mathrm{d}y_{n},X_{n-1}\in \mathrm{d}x_{n-1},Y_{n-1}\in
\mathrm{d}y_{n-1}\right) }{\mathbb{P}\left( \tau >t_{n-1},X_{n-1}\in \mathrm{%
d}x_{n-1},Y_{n}\in \mathrm{d}y_{n}\right) } \\
&=&\mathbb{P}\left( \tau >t_{n}\Big\vert\tau >t_{n-1},X_{n}\in \mathrm{d}%
x_{n},Y_{n}\in \mathrm{d}y_{n},X_{n-1}\in \mathrm{d}x_{n-1},Y_{n-1}\in
\mathrm{d}y_{n-1}\right)  \\
&&\times \mathbb{P}\left( \tau >t_{n-1},X_{n}\in \mathrm{d}x_{n},Y_{n}\in
\mathrm{d}y_{n}\Big\vert\tau >t_{n-1},X_{n-1}\in \mathrm{d}x_{n-1},Y_{n-1}\in
\mathrm{d}y_{n-1}\right)  \\
&=&\mathbb{P}\left( \tau >t_{n}\Big\vert\tau >t_{n-1},X_{n}\in \mathrm{d}%
x_{n},Y_{n}\in \mathrm{d}y_{n},X_{n-1}\in \mathrm{d}x_{n-1},Y_{n-1}\in
\mathrm{d}y_{n-1}\right)  \\
&&\times \mathbb{P}\left( X_{n}\in \mathrm{d}x_{n},Y_{n}\in \mathrm{d}%
y_{n} \Big\vert X_{n-1}\in \mathrm{d}x_{n-1},Y_{n-1}\in \mathrm{d}y_{n-1}\right) .
\end{eqnarray*}%
The last equality is true because $\left( X_{t},Y_{t}\right) _{t\geq 0}$ is Markovian. By hypothesis
\begin{eqnarray}
\lefteqn{\mathbb{P}\left( X_{n}\in \mathrm{d}x_{n},Y_{n}\in \mathrm{d}y_{n}\Big\vert X_{n-1}\in \mathrm{d}x_{n-1},Y_{n-1}\in \mathrm{d}y_{n-1}\right)}
\nonumber \\
&=&f_{\Delta Z_{n}}(x_{n}-y_{n}-(x_{n-1}-y_{n-1}))f_{\Delta
Y_{n}}(y_{n}-y_{n-1}).  \label{joint2}
\end{eqnarray}%
Then, if we take $V:=X-\mathbb{E}(X)-%
\frac{\mathrm{Cov}(X,Z)}{\mathrm{Var}(Z)}(Z-\mathbb{E}(Z))$
\begin{eqnarray}
\lefteqn{\mathbb{P}\left( \tau >t_{n}\Big\vert \tau >t_{n-1},X_{n}\in \mathrm{d}x_{n},Y_{n}\in \mathrm{d}y_{n},X_{n-1}\in \mathrm{d}x_{n-1},Y_{n-1}\in
\mathrm{d}y_{n-1}\right)}   \nonumber \\
&=&\mathbb{P}\left( \tau >t_{n}\Big\vert\tau >t_{n-1},X_{n}\in \mathrm{d}%
x_{n},V_{n}\in \mathrm{d}v_{n},X_{n-1}\in \mathrm{d}x_{n-1},V_{n-1}\in
\mathrm{d}v_{n-1}\right)   \nonumber \\
&=&\mathbb{P}\left( \tau >t_{n}\Big\vert\tau >t_{n-1},X_{n}\in \mathrm{d}%
x_{n},X_{n-1}\in \mathrm{d}x_{n-1}\right)   \nonumber \\
&=&\left(1-\exp \left\{ -\frac{2(x_{n}-c)(x_{n-1}-c)}{\sigma
_{X}^{2}(t_{n}-t_{n-1})}\right\}\right)\textbf{1}_{\{x_n>c,x_{n-1}>c\}} .  \label{joint3}
\end{eqnarray}%
The last equality is a very well known result, see for instance \cite[Lemma 2]{JeaValc}. So, from (\ref{joint1}), (\ref{joint2}) and (%
\ref{joint3}) we get the result.
\end{proof}
The densities $f_{\Delta Y_{n}}$ and $f_{\Delta Z_{n}}$ in (\ref{joint2}) can be computed in a straightforward manner. Indeed, if we define $\mu_{Y}:=\mathbb{E}[Y_{1}]$ and $\sigma_{Y}:=\mathbb{E}[(Y_{1}-\mathbb{E}[Y_{1}])^{2}]$, then the density $f_{\Delta Y_{n}}$ corresponds to a Gaussian density
with mean $\mu_{Y}(t_{n}-t_{n-1})$ and variance $\sigma_{Y}^{2}(t_{n}-t_{n-1})$,
and so
\[
f_{\Delta Y_{n}}(y)=\frac{1}{\sqrt{2\pi\sigma_{Y}^{2}(t_{n}-t_{n-1})}}\exp\left\{ -\frac{\left(y-\mu_{Y}(t_{n}-t_{n-1})\right)^{2}}{2\sigma_{Y}^{2}(t_{n}-t_{n-1})}\right\} .
\]
The expression for $f_{\Delta Z_{n}}$ is obtained analogously.

\bibliographystyle{plain}
\bibliography{references}

\begin{thebibliography}{10}

\bibitem{BlackCox}
F.~Black and J.~C. Cox.
\newblock Valuing corporate securities: some effects of bond indenture
  provisions.
\newblock {\em Journal of Finance}, 31:351--367, 1976.

\bibitem{bor12}
A.N. Borodin and P.~Salminen.
\newblock {\em Handbook of Brownian motion-facts and formulae}.
\newblock Birkh{\"a}user, 2012.

\bibitem{BrigoCoco}
D.~Brigo, J.~Garcia, and N.~Pede.
\newblock Coco bonds valuation with equity- and credit-calibrated first passage
  structural models.
\newblock {\em Working Paper, Imperial College London}, February 2013.

\bibitem{BrigoAT1P}
D.~Brigo and M.~Tarenghi.
\newblock Credit defaulty swap calibration and equity swap valuation under
  counterparty risk with a tractable structural model.
\newblock In {\em Proceedings of the FEA 2004 Conference at MIT}, 2004.

\bibitem{BryisDeVarenne}
E.~Bryis and F.~{de Varenne}.
\newblock Valuing fixed rate debt: An extension.
\newblock {\em Journal of Financial and Quantitave Analysis}, 32:329--248,
  1997.

\bibitem{CocGemJea}
D.~Coculescu, H.~Geman, and M.~Jeanblanc.
\newblock Valuation of default-sensitive claims under imperfect information.
\newblock {\em Finance Stoch.}, 12:195--218, 2008.

\bibitem{CDufGolHel}
P.~Collin-Dufresne, R.~Goldstein, and J.~Helwege.
\newblock P. is credit event risk priced? modeling contagion via the updating
  of beliefs.
\newblock {\em Working paper, Carnegie Mellon University}, 2003.

\bibitem{Cocacoco}
J.~M. Corcuera, J.~De Spiegeleer, J.~Fajardo, H.~J{\"o}nsson, W.~Schoutens, and
  A.~Valdivia.
\newblock Close form pricing formulas for coupon cancellable cocos.
\newblock {\em Journal of Banking \& Finance}, 42:339--351, May 2014.

\bibitem{DuffieLando}
D.~Duffie and D.~Lando.
\newblock Term structure of credit spreads with incomplete accounting
  information.
\newblock {\em Econometrica}, 69:633--664, 2001.

\bibitem{GouJarZen}
X.~Guo, R.~A. Jarrow, and Y.~Zeng.
\newblock Credit risk models with incomplete information.
\newblock {\em Mathematics of Operations Research}, 34(2):320--332, 2009.

\bibitem{HieScherer}
P.~Hieber and M.~Scherer.
\newblock A note on first-passage times of continuously time-changed brownian
  motion.
\newblock {\em Statistics and Probability Letters}, 82:165--172, 2012.

\bibitem{JeaValc}
M.~Jeanblanc and S.~Valchev.
\newblock Partial information and hazard process.
\newblock {\em International Journal of Theoretical and Applied Finance},
  8:807--838, 2005.

\bibitem{LoLeeHui}
C.~F. Lo, H.~C. Lee, and C.~H. Hui.
\newblock A simple approach for pricing barrier options with time-dependent
  parameters.
\newblock {\em Quantitative Finance}, 3:98--107, 2003.

\bibitem{LongstaffSchwartz}
F.~A. Longstaff and E.~S. Schwartz.
\newblock A simple approach to valuing risky fixed and floating rate debt.
\newblock {\em Journal of Finance}, 50:789--819, 1993.

\bibitem{Merton}
R.~C. Merton.
\newblock On the pricing of corporate debt: The risk structure of interest
  rates.
\newblock {\em J. Finance}, 29(2):449--470, 1974.

\bibitem{MetAti}
S.~Metwally and A.~Atiya.
\newblock Using brownian bridge for fast simulation of jump-diffusion processes
  and barrier options.
\newblock {\em Journal of Derivatives}, 10:43--54, 2002.

\bibitem{Rapisarda}
F.~Rapisarda.
\newblock Barrier options on underlyings with time--dependent parameters: a
  perturbation expansion approach.
\newblock {\em Technical Report, Product and Business Development Group, Banca
  IMI}, May 2005.

\bibitem{RufScherer}
J.~Ruf and M.~Scherer.
\newblock Pricing corporate bonds in an arbitrary jump-diffusion model based on
  an improved brownian-bridge algorithm.
\newblock {\em Journal of Computational Finance}, 2011.

\bibitem{SaaRaquejoSantaClara}
J.~Sa{\'a}-Raquejo and P.~Santa-Clara.
\newblock Bond pricing with default risk.
\newblock {\em Working paper, UCLA}, 1999.

\bibitem{CetinEtAl}
R.~Jarrow P.~Protter {U. \c{C}etin} and Y.~Yildirim.
\newblock Modelling credit risk with incomplete accounting information.
\newblock {\em Annals of applied probability}, 14(3):1167--1178, 2004.

\bibitem{zhou2001term}
Chunsheng Zhou.
\newblock The term structure of credit spreads with jump risk.
\newblock {\em Journal of Banking \& Finance}, 25(11):2015--2040, 2001.

\end{thebibliography}

\end{document}